\documentclass[secnumarabic,amssymb,nobibnotes,aps,pre,
twocolumn,
floatfix]{revtex4-2}
\usepackage[pdftex]{graphicx}
\usepackage{amsmath, amssymb, amsthm}
\usepackage{braket}
\usepackage{natbib}
\usepackage{color}
\usepackage{lipsum}
\usepackage{comment}
\usepackage{qcircuit}
\usepackage{physics}
\usepackage[ruled]{algorithm2e}
\usepackage{subcaption}
\usepackage[english]{babel}
\newtheorem{theorem}{Theorem}
\theoremstyle{definition}
\newtheorem{definition}{Definition}[section]
\usepackage{xcolor}
\usepackage[hidelinks,colorlinks=true,linkcolor=blue, citecolor=blue]{hyperref}
\usepackage{cleveref}

\begin{document}
\title{Towards enhancing quantum expectation estimation of matrices through partial Pauli decomposition techniques and post-processing}
\author{Lu Dingjie}
\email[]{ludj@ihpc.a-star.edu.sg}
\affiliation{Institute of High Performance Computing, Agency for Science, Technology and Research (A*STAR), 138632 Singapore}

\author{Li Yangfan}
\affiliation{Institute of High Performance Computing, Agency for Science, Technology and Research (A*STAR), 138632 Singapore}

\author{Dax Enshan Koh}
\affiliation{Institute of High Performance Computing, Agency for Science, Technology and Research (A*STAR), 138632 Singapore}

\author{Wang Zhao}
\affiliation{Institute of High Performance Computing, Agency for Science, Technology and Research (A*STAR), 138632 Singapore}

\author{Liu Jun}
\affiliation{Institute of High Performance Computing, Agency for Science, Technology and Research (A*STAR), 138632 Singapore}

\author{Liu Zhuangjian}
\email[]{liuzj@ihpc.a-star.edu.sg}
\affiliation{Institute of High Performance Computing, Agency for Science, Technology and Research (A*STAR), 138632 Singapore}

\begin{abstract}
We introduce an approach for estimating the expectation values of arbitrary
$n$-qubit matrices $M \in \mathbb{C}^{2^n\times 2^n}$ on a quantum computer. In
contrast to conventional methods like the Pauli decomposition that utilize $4^n$
distinct quantum circuits for this task, our technique employs at most $2^{n+1}$ unique circuits, with even fewer required for matrices with limited bandwidth. Termed the \textit{partial Pauli decomposition}, our method involves observables formed as the Kronecker product of a single-qubit Pauli operator and orthogonal projections onto the computational basis. By measuring each such observable, one can simultaneously glean information about $2^n$ distinct entries of $M$ through post-processing of the measurement counts. This reduction in quantum resources is especially crucial in the current noisy intermediate-scale quantum era, offering the potential to accelerate quantum algorithms that rely heavily on expectation estimation, such as the variational quantum eigensolver and the quantum approximate optimization algorithm.
\end{abstract}

\maketitle

\section{Introduction}\label{sec:1_intro}
Noisy intermediate-scale quantum (NISQ) devices represent a significant milestone in the pursuit of fault-tolerant quantum computers \citep{Preskill2018Quantum}. Characterized by qubit counts ranging from approximately 50 to several hundred \citep{Tannu2019}, these NISQ devices grapple with short coherence times, limited qubit connectivity, and vulnerability to noise \citep{Ravi2023}.
Designed to address these limitations, a noteworthy category of algorithms that has gained increased attention is the family of variational quantum algorithms (VQAs) \citep{Peruzzo2014NC,Cerezo2021NRP}, which operate by formulating a cost function in terms of the expectation values of quantum observables. These values are estimated through measurements of the output states of parameterized quantum circuits, with parameters optimized by a classical optimizer.
Notably, VQAs like the variational quantum eigensolver (VQE)
\citep{Peruzzo2014NC} and the quantum approximate optimization algorithm (QAOA)
\citep{Farhi2014arXiv} have proven highly versatile, finding extensive
applications across diverse fields, including chemistry
\citep{Lee2022JCTC,Singh2023JCP}, quantum many-body systems
\citep{you2021exploring}, linear systems
\citep{Xu2021SB,Huang2021NJP,Liu2021PRA,Sato2021PRA}, subatomic physics
\citep{Lu2019PRA}, electromagnetics \citep{Ewe2021IEEE}, fluid dynamics
\citep{Leong2022SR}, colloidal transport \citep{Leong2023IJNMHFF}, and topology optimization \citep{Sato2022arXiv}.

A central challenge for VQAs revolves around the computational costs associated with estimating expectation values of matrices \citep{Wecker2015PRA}. In many scenarios, direct acquisition of these expectations through the measurement of a single observable proves inefficient, prompting the adoption of indirect techniques. When devising such techniques, two primary considerations come to the forefront. The first consideration revolves around the number of terms involved in decomposing the matrix in terms of easily measurable observables. A common approach is the Pauli decomposition, wherein matrices are expressed as a linear combination of multi-qubit Pauli operators. Specifically, an $n$-qubit matrix $A \in \mathbb C^{2^n \times 2^n}$ is expressed as $A = \sum_i \alpha_i P_i$, where each $P_i$ is a Pauli operator, and $\alpha_i = \tr(P_i A)/2^n$ is the coefficient of $P_i$ in the Pauli decomposition of $A$. However, this approach is often inefficient for many pertinent problems, potentially resulting in as many as $4^n$ non-zero terms. Moreover, even in the simple scenario where the matrix $A$ has only one non-zero element, this decomposition may entails a total of $2^n$ non-zero terms. This inefficiency is especially evident in problems involving numerical methods, where band-width matrices are utilized. For instance, for the finite element method (FEM), a widely used numerical technique that involves discretizing continuous physical domains into a finite number of elements connected with nodes, the corresponding matrix $A$ in the system of linear equations $Ax=b$---with $A\in {\mathbb{C}^{2^n\times 2^n}}$ being the stiffness matrix, $b$ the external loading vector, and $x$ the solution---typically possesses $O(4^n)$ non-zero components in the Pauli decomposition.

Alternative schemes for representing $A$ include expressing it as a weighted sum of Pauli strings \citep{Huang2021NJP,Bravo2019arX,Xu2021SB}, or in terms of quantum-compatible operators that may deviate from being Pauli (or more generally, Hermitian) \citep{Liu2021PRA,Sato2021PRA}. However, these approaches often exhibit a tendency to be problem-specific. Furthermore, linear combinations or spectral decompositions prescribed by existing methods may not be well-suited for NISQ devices, given the limitations imposed by their circuit depth and qubit connectivity.

The second consideration is the measurement workload, which is determined by the number of unique circuits used in the estimation process \citep{Jena2019arXiv,Izmaylov2019JCTC,Zhao2020PRA}.
As stated above, a naïve approach of expectation evaluation is to decompose the matrix $A$ into a linear combination of Pauli operators as $A=\sum_{i=1}^L\alpha_i P_i$, where for each $i\in \{1,\ldots,L\}$, $\alpha_i$ is a complex-valued coefficient and $P_i$ is Pauli operator~\citep{Peruzzo2014NC}. This process is followed by separate measurements of each term. While conceptually straightforward, this approach is computationally expensive because of the numerous individual measurements required. The number $L$ of measurements can be as large as $4^n$ when $A$ is a fully dense matrix; explicitly, one can write the expectation as $\bra \psi A\ket\psi= \sum_{i=1}^L c_i \bra{\psi} P_i\ket{\psi} $. To increase efficiency, an improved strategy that has been proposed involves clustering mutually commuting Pauli operators \citep{Gokhale2020IEEETQE,Kandala2017NATURE,Verteletskyi2020TJCP}, thereby enabling simultaneous measurements on all the operators within each cluster.
However, minimizing the number of clusters may not necessarily reduce the total number of measurements \citep{McClean2016NJP}. The grouping is efficiently optimized using a minimal clique cover, which clusters qubit-wise commuting terms \citep{Verteletskyi2020TJCP}. However, the total number of groups increases with the number of Pauli operators, scaling as $O(n^4)$, making it less practical for larger-scale problems. Another measurement reduction method involves the classical shadow technique which acquires a classical description of a quantum state through subjecting it to a random unitary before measuring it \citep{Huang2020NP,Koh2022Quantum,Bu2024npjQI}.
However, this method too faces challenges, including issues with the selection of suitable ensembles, and the practical implementation of random measurements. Yet another strategy is based on Bell measurements \citep{Liu2021PRA,Kondo2022Quantum,Sato2022arXiv}; such a strategy can optimize the circuit count and is expected to be effective for matrices with limited band-widths.

In response to these two considerations, we have proposed an efficient algorithm to transform arbitrary matrices into a readily comprehensible and highly implementable form with the \textit{partial Pauli decomposition}. This transformation utilizes a sequence of CNOT gates to map matrices onto the Pauli basis. In addition, this approach provides the capability to generate $2^{n}$ distinct clusters for an $n$-qubit system, each comprising $2^{n}$ terms positioned at specific locations within the $2^n \times 2^n$ matrix. This grouping enables simultaneous measurements for each cluster, allow for the prediction of multiple properties via post-processing. Furthermore, this method ultimately results in a complete matrix representation that allows a greater degree of simultaneity in measurements compared to sequential measurements of each term separately. Notably, binary arithmetic, in conjunction with the partial Pauli measurements, facilitates measurements with fewer computational resources, e.g.\ qubits, and offers a clear and easily interpretable strategy for implementation.

For better clarity, we outline the merits of the proposed algorithm for arbitrary $2^n \times 2^n$ matrices corresponding to $n$-qubit systems.

\begin{itemize}
    \item[1.]\textbf{Efficiently transforming arbitrary matrices into a quantum-compatible form:} The method entails $2^n$ unique transformation matrices, each corresponding to a specific circuit, accommodating $2^{n}$ distinct components. Consequently, these matrices represent arbitrary $2^n \times 2^n$ matrices. Furthermore, this transformation exclusively comprises only CNOT gates, acknowledged for their efficiency in implementation and their controllability.
    \item[2.] \textbf{Efficient decomposition with the partial Pauli decomposition:} In contrast to the conventional Pauli decomposition, a subset of $m$ qubits, instead of the full $n$ qubits, requires transformation into the Pauli basis. The rest of the $n-m$ qubits are left as free qubits, represented by 0's and 1's. The reduction in qubit count decreases the constraint set size for reconstructing the matrix from $2^n$ to $2^m$; when $m=1$, the size of this set is just 2. Consequently, this exponentially reduces the number of constraints for reconstructing the matrix in a divide-and-conquer fashion, thereby exponentially reducing the circuit count.
    \item[3.] \textbf{Efficient deterministic O(1) grouping:} The grouping of $2^n\times2^n$ terms in a fully dense matrix into a total of $2^n$ clusters can be achieved using a standardized O(1) process. This streamlined clustering enables the simultaneity of measurements, consequently diminishing the computational overhead.
    \item[4.] \textbf{Well-defined and efficient circuit construction:} The binary-arithmetic-based transformation identification and circuit construction strategy enhances the efficiency by eliminating the iterative process in~\citep{Kondo2022Quantum}, resulting in a complexity reduction to $O(1)$. The appendix will detail this simplification. Moreover, classical techniques can post-process the information stored in the aforementioned $2^{n-m}$ qubits, which proves more efficient than using quantum algorithms.
    \item[5.] \textbf{Efficient measurement:} The algorithm's design allows measurements at any bit, wherein one measurement covers $2^{n}$ terms, thereby extending its applicability beyond specific locations, and makes it adaptive to the connectivity constraints in NISQ quantum hardware.

\end{itemize}
As a result, our method effectively improves efficiency by addressing the two aforementioned considerations. The remainder of the paper is structured as follows. Section~\ref{sec:2_expe_esti} delineates our strategy of expectation estimation. Section~\ref{sec:3_method} describes the proposed algorithm and the corresponding implementation. Section~\ref{sec:4_complexity} presents the computational complexity of the proposed method in terms of the number of unique circuits and the time complexity.
Section~\ref{sec:5_results} offers the results and the corresponding discussion. Last but not least, Section~\ref{sec:6_conclusion} concludes our main findings, identifies limitations, and suggests avenues for future research.

\section{Estimation of $\bra{\phi} M \ket{\psi}$}\label{sec:2_expe_esti}
The objective of this work is to propose and implement a protocol to compute inner products of the form $\bra{\phi} M \ket{\psi}$ on a quantum computer. Specifically, let $\ket{\phi}, \ket{\psi} \in {\mathbb C}^{2^n}$ be two $n$-qubit quantum states, and $M \in {\mathbb C}^{2^n\times 2^n}$ be an arbitrary $n$-qubit matrix, which we write explicitly as
\begin{align}
\label{eq:exp}
    M = \begin{bmatrix}
    M_{00} & M_{01} &\cdots & M_{0,2^n-2} & M_{0,2^n-1} \\
    M_{10} & M_{11} &\cdots & M_{1,2^n-2} & M_{1,2^n-1} \\
    \vdots&\vdots&\ddots&\vdots&\vdots\\
    M_{2^n-2,0} & M_{2^n-2,1} &\cdots& M_{2^n-2,2^n-2} & M_{2^n-2,2^n-1} \\
    M_{2^n-1,0} & M_{2^n-1,1} &\cdots& M_{2^n-1,2^n-2} & M_{2^n-1,2^n-1}
\end{bmatrix}
\end{align}
where $M_{rc}$, for $r, c \in \{0,1, \dots, 2^n-1\}$, is the entry occupying the $r$-th row and $c$-th column of matrix $M$.
This inner product can be expressed as
\begin{eqnarray}
    \bra{\phi} M \ket{\psi} = \sum^{2^n-1}_{r, c=0}M_{rc}\bra{\phi}\ket{r}\bra{c}\ket{\psi}.
\end{eqnarray}
For simplicity, we consider the case where $\phi = \psi$ for the following discussion. The case of $\phi \neq \psi$ can be transformed into the present case with the same states by introducing an additional ancillary bit, as we demonstrate in \Cref{app:two_states}.

It is straightforward to estimate the expectation values corresponding to the diagonal entries of $M$ directly, namely by measuring the state $\phi$ and computing the sum
\begin{eqnarray}
    \sum^{2^n-1}_{i=0}M_{ii}\bra{\phi}\ket{i}\bra{i}\ket{\phi}
    =\sum^{2^n-1}_{i=0}M_{ii}|\bra{i}\ket{\phi}|^2,
\end{eqnarray}
where $|\bra{i}\ket{\phi}|^2$ is the probability of obtaining the bit-string $i$ when directly measuring $\ket{\phi}$ in the computational basis.

The challenge here resides in the estimation of the expectation values corresponding to off-diagonal terms:
\begin{eqnarray}
    \sum^{2^n-1}_{r=0}\sum^{2^n-1}_{c=0, r \neq c}M_{rc}\bra{\phi}\ket{r}\bra{c}\ket{\phi}.
\end{eqnarray}
As mentioned in the introduction, one established method is to express the $2^n$ off-diagonal terms in terms of the tensor product of $n$ Pauli matrices.

Specifically, for $n$-qubit system, in the naive Pauli decomposition, the expectation estimation with Pauli measurement transforms all qubits of the target state $\ket{\psi}$ into the Pauli bases, the observable of this Pauli measurement corresponding to a matrix $\mathcal P_n
= \big\{ P_1 \otimes \cdots \otimes P_n : \forall i,\ P_i \in \{I, X, Y,
Z\}\big\}$ and the Pauli matrices~\cite{Nielsen2010QCQI} are defined by
\begin{align} \label{eq:pauli_gates}
    I&=\ket{0}\!\bra{0}+\ket{1}\!\bra{1},\
    X=\ket{1}\!\bra{0}+\ket{0}\!\bra{1}, \nonumber\\
    Z&=\ket{0}\!\bra{0}-\ket{1}\!\bra{1},\ Y=iXZ.
\end{align}
The expectation $\bra{\psi} \mathcal P_n \ket{\psi}$ involves information with the bit string determined by the Pauli string.

The matrix $M$ can be decomposed into a linear combination of multiple elements of $\mathcal P_n$. The challenge of this approach lies in the fact that $\mathcal P_n$ forms an orthogonal basis for the vector space of $2^n \times 2^n$ complex matrices. The matrix $M$ is to be expressed as the linear combination of this basis, and $4^n$ terms are involved in the linear combination in the worst case. Meanwhile, determining the coefficient requires multiplying two $2^n \times 2^n$ matrix, rendering the naive Pauli decomposition inefficient.

Here we propose the strategy of partition $M$ and decomposed the partitions in the subspace with lower dimension. Specifically, when $m$ qubits, initially in state $\ket{\phi}$, undergo a transformation to the Pauli bases (where $m < n$), the measurement involves only these $m$ qubits, corresponding to the length-$m$ Pauli operators. Consequently, the remaining $n-m$ qubits are free to adopt the outcomes of 0 or 1, which corresponding to the Hermitian $\ket{0} \bra{0}$ and $\ket{1} \bra{1}$. With the length $n-m$ bits determined and the other $m$ qubits form the length-$m$ Pauli operators, a sub-space of $4^{n-m}$. And in this subspace the worst case for the linear combination is bounded by $4^{n-m}$ terms. Since the $n-m$ bits are free there are $2^{n-m}$ such sub-space and the matrix $M$ is also partitioned into $2^{n-m}$ parts. However, it is crucial to note that the information for the $n-m$ bits is derived via post-processing measured bit counts, leading to the information boost from 1 piece in the measurement of length-$n$ Pauli operator to $2^{n-m}$ pieces in the measurement of length-$m$ Pauli operator. Following this strategy, one way to enhance the Naive Pauli decomposition is by minimizing the dimension of the subspace to minimize the bound on the worst case and resulting in $m=1$, which is the subspace defined by ${I, X, Y, Z}$ the famous Pauli bases for $2 \times 2$ matrix. This approach involves separately computing the expectation $\bra{\phi}\ket{r}\bra{c}\ket{\phi}$ and addressing the coefficient $M_{rc}$ afterward, which is termed as \textit{partial Pauli decomposition} in this context.

Therefore, the task becomes finding efficient methods to estimate the values of
\begin{eqnarray}
    \bra{\phi}\ket{r}\bra{c}\ket{\phi}, r \neq c,
    \label{eq:rc}
\end{eqnarray}
with one qubit been measured. Here $(r, c)$ are the row, column index pair, and $\ket{r} \bra{c}$ is the $2^n \times 2^n$ matrix with all 0s but 1 at row $r$, column $c$.
Generally, quantum expectations $\bra{\phi}\ket{r}\bra{c}\ket{\phi}$ don't directly translate to measurable observables. However, if a suitable unitary transformation $T$ can be identified to map the pair $(r, c)$ to a specific $( \bar r, \bar c)$ corresponding to a measurable observable in a new quantum state, estimating the expectation in Eq.~\ref{eq:rc} becomes feasible. Consequently, it allows the determination of the value in Eq.~\ref{eq:exp} as
\begin{align}
    \bra{\phi} \ket{r}\bra{c} \ket{\psi} &= \bra{\phi} T^T T \ket{r}\bra{c} T^T T \ket{\psi} \nonumber \\
    &= \bra{\phi} T^T \ket{\bar r}\bra{\bar c} T \ket{\psi} =\bra{\phi'}\ket{\bar r}\bra{\bar c}\ket{\psi'},
    \label{eq:T}
\end{align}
where $T$ is the unitary transformation, $\ket{\psi'} = T \ket{\psi}$, $\ket{\phi'} = T \ket{\phi}$ are the new quantum states, $\ket{\bar r} = T \ket{r}$, $\ket{\bar c} = T \ket{c}$ are the new row and column, and $\ket{\bar r} \bra{\bar c}$ is a new matrix. Given our objective to transform a single qubit into the Pauli basis, the resulting states $\ket{\bar r}$ and $\ket{\bar c}$ should differ by only a single bit, thereby implying the single-bit differing states defined in Sec.~\ref{sec:SBDS}.

\section{Method}\label{sec:3_method}
\subsection{The single-bit differing states}\label{sec:SBDS}
To efficiently calculate the expectation in Sec.~\ref{sec:2_expe_esti}, the single-bit differing states, denoted as the states pair as $(\ket{r}, \ket{r'})$ are proposed. For $n$-qubit system, the single-bit differing states $(\ket{r}, \ket{r'})$ is defined as the states $\ket{r}$ and $\ket{r'}$, such that there are only a single differing bit,i.e. the Hamming weight $\mathrm{wt} (r \oplus r')=1 $.
\begin{eqnarray}
    \label{eq:rr}
    \ket{r} &=& \ket{b_1\ldots b_{k-1} \ 0 \ b_{k+1} \ldots b_n}, \\ \nonumber
    \ket{r'} &=& \ket{b_1\ldots b_{k-1} \ 1 \ b_{k+1} \ldots b_n},
\end{eqnarray}
with $b_i \in \mathbb Z_2,i \in \{1, 2, \ldots n-1, n\}$. $\ket{r}$ and $\ket{r'}$ only differ in the $k$-th bit. The single-bit differing states can be efficiently calculated from arbitrary $n$-qubit pair $(\ket{r}, \ket{c})$ with binary arithmetic introduced in Sec.~\ref{bina_arit}

\subsection{Binary arithmetic and matrix permutation}\label{bina_arit}
The outcome of applying the $X$ and $CNOT$ gates to an $n$-qubit state can be efficiently represented using binary arithmetic, as the effect of $X$ and $CNOT$ gates are illustrated in Fig.~\ref{fig:circ}.
When an $X$ gate is applied to a qubit $q_0$ in the state $a$, the resulting state $X(q_0)$ becomes $a \oplus 1$. Similarly, when a $CNOT$ gate is applied with $q_1$ as the control qubit at state $b$ and $q_2$ as the target qubit at state $c$, the resulting states are $b$ and $b \oplus c$ for $q_1$ and $q_2$, respectively.
Here $a, b, c \in \mathbf Z^2$, and the operator $\oplus$ is the binary addition mod 2.

When a series of $X$ and $CNOT$ gates is applied, they collectively represent a transformation $T$, which is denoted as a matrix $P$ of size $2^n \times 2^n$. Applying $P$ to the $2^n \times 2^n$ matrix $A$ through the outer product expression $A = \ketbra{r}{c}$
 where $r, c \in {1, 2, \ldots, n-1, n}$, yields the resulting matrix $\bar A = PAP^{\dagger} = P\ketbra{r}{c}P^{\dagger} = \ketbra{\bar r}{\bar c}$, where $\ket{\bar r} = P \ket{r}$ and $\ket{\bar c} = P \ket{c}$.  This permutation operation of $P$ on $A$, or specifically the effect of consecutive $X$ and $CNOT$ gates, relocates the unity entry at row $r$, column $c$ in $A$ to a new position denoted by row $\bar r$, column $\bar c$. This permutation on the bit-string of $r$ and $c$ can be expressed as the linear system in $\mathbb Z^2$ as,
\begin{eqnarray}
    \label{eq:bin_bc}
    \bar r &=& T r \oplus R, \\ \nonumber
    \bar c &=& T c \oplus R,
\end{eqnarray}
where $r, \bar r, c, \bar c, R$ are binary vectors of length $n$ representing their binary forms. $R$ is a constant, $T$ is an $n \times n$ invertible matrix since the transformation derived from $X$ and CNOT gates is unitary.

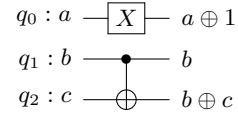
\begin{figure}
\centering
\begin{equation*}
\Qcircuit @C=1.em @R=0.4em @!R {
    \lstick{q_0: a}   & \gate{X} &  \rstick{a \oplus 1} \qw \\
    \lstick{q_1: b}   & \ctrl{1} &  \rstick{b} \qw \\
    \lstick{q_2: c}   & \targ &  \rstick{b \oplus c} \qw
}
\end{equation*}
\caption{The state before and after applying the $X$ and $CNOT$ gates on computational basis states}
\label{fig:circ}
\end{figure}

\subsection{Identify the circuits of $T$}\label{sec:circ}
The single-bit differing states introduced in Sec.~\ref{sec:SBDS} are constructed from the binary transformation $T$, which turns $(r, c)$ to $(r, r')$ based on Eq.~\ref{eq:bin_bc} as
\begin{subequations} \label{eq:TT}
\begin{eqnarray}
    r &=& T r \oplus R, \label{eq:rTr}\\
    r' &=& T c \oplus R. \label{eq:rTc}
\end{eqnarray}
\end{subequations}
It is straightforward to get the constant binary vector $R=0$ from Eq.~\ref{eq:rTr}, then Eq.~\ref{eq:TT} is simplified as
\begin{eqnarray}
    r &=& T r, \\ \nonumber
    r' &=& T c.
\end{eqnarray}
Over the ring $\mathbb Z^2$, the matrix $T$ in Eq.~\ref{eq:TT} might possess various solutions. However, there exists at least one valid $T$ that can be readily formulated. To generate matrix $T$ for given values of $r$ and $c$, one method involves initializing it as an Identity matrix. This matrix is modified by replacing column $k$ with $r \oplus c$. The column $k$ is identified as the bit index where $b_k^r \oplus b_k^c = 1$, here $b_i^j$ is the $i$-th bit of $j$. A more detailed proof is provided in \Cref{app:TG}. The transformation $T$ can encompass $2^n$ row-column pairs $(p,q)$ as long as they satisfy $p \oplus q=r
\oplus c$. This condition holds true because, given a fixed value of $r \oplus c$, there are $n$ qubits available to independently represent the values of 0 or 1 for $p$ and $q$.
There are precisely $2^n - 1$ distinct transformations $T$ in the $n$-qubit case. This count arises from the fact that the sum $r \oplus c$ represents an $n$-bit number, capable of taking any value in the range $\{1, 2, \ldots 2^n-1\}$, where 0 is excluded, due to $r=c$. Hence, there are $2^n-1$ unique transformations $T$. Each transformation acts on $2^n$ elements in $2^n \times 2^n$ matrix, which is a total number of $2^n \times(2^n-1)$, exactly the amount of off-diagonal elements. Therefore, The upper bound of distinct quantum circuits to compute the entire $2^n \times 2^n$ matrix is $2^n$, covering the diagonal terms measured directly, and $2^n - 1$ for the off-diagonal terms. The important properties of non-overlap of such transformation hold and the proofs are listed in \Cref{app:prov}.

\begin{theorem}[Non-Overlap]
Let
$i, j, k, l, p, p', q, q'\in \mathbb Z_2^n$, $p \neq p'$, $q \neq q'$,
$T_p, T_q \in \mathbb Z_2^{n \times n}$,
$T_p p = p'$,
$T_q q = q'$,
$S_p = \big \{ (i, j): T_p (i, j)^T = (i, i')^T, i \oplus i' = p' \big \}$,
$S_q = \big \{ (k, l): T_q (k, l)^T = (k, k')^T, k \oplus k' = q' \big \}$. Then, $S_p \cap S_q = \emptyset $ if $p \neq q$.
\label{th:nonoverlap}
\end{theorem}

Theorem~\ref{th:nonoverlap} ensures that for these $2^n-1$ transformations, each transfers $2^n$ unique pairs to the single-bit differing states. The $2^{n-1}(2^n-1)$ independent off-diagonal terms in $M$ can be precisely accounted for by the $2^n-1$ transformations.

A dedicated algorithm is proposed as illustrated in Algorithm~\ref{alg:T} of \Cref{app:TG} to get $T$. The gates sequence for the circuit is recovered from the $T$ with Algorithm~\ref{alg:G} of \Cref{app:TG}: .

\subsection{The Partial Pauli Measurement (PPM) }\label{sec:PPM}
To maximize quantum circuit efficiency for expectation estimation, we extract a minimal pair of information by setting $m=1$ in the strategy introduced in ~\ref{sec:2_expe_esti}. This involves transforming a single qubit (the $k$-th qubit) into Pauli bases, denoted as partial Pauli measurement (PPM) in this work. There are $2^{n-1}$ different observables as the tensor product of
\begin{eqnarray}
    \mathcal M_{r,r'} = C_1^j \otimes \ldots C_{k-1}^j\otimes X \otimes C_{k+1}^j \ldots \otimes C_n^j,
\end{eqnarray}
where $C_i^j= I_j$, with $i \in \{1, 2, \ldots n-1, n\}$, $j \in \{0, 1\}$, $I_0 = \ketbra{0}{0}, I_1=\ketbra{1}{1}$. Each observable $\mathcal M_{r,r'}$ is a $2^n \times 2^n$ matrix with two non-zero off-diagonal components located at the symmetrical locations $(r,r')$ and $(r',r)$ where $r, r'$ are the single-bit differing states defined in Sec.~\ref{sec:SBDS}. Hence, the single-bit differing states directly enable the Partial Pauli Measurement by transforming the $k$-th bit into the Pauli bases and performing a measurement.

Generally, when it comes to arbitrary off-diagonal entries with row $r$ and column $c$, any qubit can be selected to transform to Pauli bases. However, given that $r \neq c$, at least one qubit indexed as $k$ will have different states, i.e. $b_k^r \oplus b_k^c = 1$.

This $k$-th bit is chosen for the maximum efficiency in the current work. If the matrix $M$ is symmetric, PPM transforms the $k$-th qubit to Pauli $X$ while if the matrix is anti-symmetric, PPM transforms the $k$-th qubit to Pauli $Y$. For the more general case, $M$ can be decomposed into symmetric $M_S$ and anti-symmetric $M_{AS}$ parts.

\subsection{Direct measurement with Hadamard Gate}\label{sec:DM}
For the states $\ket{r}$ and $\ket{r'}$ in Eq.~\ref{eq:rr}, applying the Hadamard gate $H$ on the $k$-th qubit, denoted as $H_k$, results in the superposition of the state $\ket{r}$ and the state $\ket{r'}$ as
\begin{eqnarray}
\label{eq:Hrr}
    H_k \ket{r} = \frac{1}{\sqrt{2}}(\ket{r} + \ket{r'}),\\ \nonumber
    H_k \ket{r'} = \frac{1}{\sqrt{2}}(\ket{r} - \ket{r'}).
\end{eqnarray}
As Eq.~\ref{eq:Hrr} generates both the state $\ket{r}$ and $\ket{r'}$ with opposite signs for $\ket{r'}$. These two identities serve as an effective method to transform off-diagonal terms into diagonal ones for direct measurement.
If the Hadamard gate $H_k$ is applied on the observable $\ket{r}\bra{r'} + \ket{r'}\bra{r}$, we obtain
\begin{align}
    \label{eq:H}
    & H_k(\ket{r}\bra{r'} + \ket{r'}\bra{r})H_k^{\dagger}
    \nonumber\\
    &= H_k\ket{r}\bra{r'}H_k^{\dagger} + H_k\ket{r'}\bra{r})H_k^{\dagger} \\ \nonumber
    &= \frac{1}{\sqrt{2}} (\ket{r} + \ket{r'}) \frac{1}{\sqrt{2}} (\bra{r} - \bra{r'}) \\ \nonumber
    &\qquad+ \frac{1}{\sqrt{2}} (\ket{r} - \ket{r'}) \frac{1}{\sqrt{2}} (\bra{r} + \bra{r'}) \\ \nonumber
    &= \frac{1}{2} (\ket{r}\bra{r} - \ket{r}\bra{r'} + \ket{r'}\bra{r} - \ket{r'}\bra{r'}) \\ \nonumber
    &\qquad+ \frac{1}{2} (\ket{r}\bra{r} + \ket{r}\bra{r'} - \ket{r'}\bra{r} - \ket{r'}\bra{r'})\\ \nonumber
    &= \ket{r}\bra{r} - \ket{r'}\bra{r'}.
\end{align}
This process demonstrates the transformation of off-diagonal terms to the diagonal ones. Applying the $S$ gate on the $k$-th qubit and utilizing the identities
\begin{align}
    \label{eq:S}
    S_k\ket{r} &= \ket{r}\\ \nonumber
    S_k\ket{r'} &= i\ket{r'}\\ \nonumber
    \bra{r}S_k^{\dagger} &= \bra{r}\\ \nonumber
    \bra{r'}S_k^{\dagger} &= -i\bra{r'}. \nonumber,
\end{align}
the observable expressed as $\ket{r}\bra{r'} - \ket{r'}\bra{r}$ can be transformed to
\begin{eqnarray}
   S_k(\ket{r} \bra{r'} -\ket{r'} \bra{r})S_k^{\dagger} = -i(\ket{r} \bra{r'} +\ket{r'} \bra{r}),
\end{eqnarray}
which resembles Eq.~\ref{eq:H} except for an additional complex coefficient $i$.

As the single-bit differing states $(\ket{r}, \ket{r'})$ are constructed, the expectation of $\bra{\phi}\ket{r}\bra{c}\ket{\phi}$ can be calculated.

Since $\bra{\phi}\ket{r}\bra{c}\ket{\phi}$ is in general a complex number, it is of the form
\begin{eqnarray}
    \bra{\phi} \ket{r} \bra{c} \ket{\phi} = \Re(\bra{\phi} \ket{r} \bra{c} \ket{\phi}) + i \Im(\bra{\phi} \ket{r} \bra{c} \ket{\phi}).
\end{eqnarray}
Here $\Re(\bra{\phi} \ket{r} \bra{c} \ket{\phi})$ and $\Im(\bra{\phi} \ket{r} \bra{c} \ket{\phi})$ are both real numbers representing the real and imaginary parts of the complex number $\bra{\phi}\ket{r}\bra{c}\ket{\phi}$. More importantly, the real and imaginary parts of $a^*b, a, b \in \mathbb C$ can expressed as
\begin{eqnarray}
   \Re(a^*b) &=& \frac{1}{2} (a^*b + b^*a), \\ \nonumber
   \Im(a^*b) &=& -\frac{i}{2} (a^*b - b^*a).
\end{eqnarray}
Hence, after constructing single-bit differing states by $T$, $\Re(\bra{\phi} \ket{r} \bra{c} \ket{\phi})$ and $\Im(\bra{\phi} \ket{r} \bra{c} \ket{\phi})$ can be calculated via the relation introduced in Eqs.~\ref{eq:H}--\ref{eq:S} as,
\begin{align}
\label{eq:ReIm}
    &\Re(\bra{\phi}\ket{r} \bra{c} \ket{\phi}) \nonumber\\
    &= \frac{1}{2} (\bra{\phi} T^{\dagger} H_k^{\dagger} H_k T\ket{r} \bra{c} T^{\dagger} H_k^{\dagger} H_k T \ket{\phi}  \nonumber\\
    &\quad+ \bra{\phi} T^{\dagger} H_k^{\dagger} H_k T \ket{c} \bra{r} T^{\dagger} H_k^{\dagger} H_k T \ket{\phi})   \nonumber\\
    &= \frac{1}{2} (\bra{\phi'}\ket{r} \bra{r'} \ket{\phi'} + \bra{\phi'} \ket{r'} \bra{r}  \ket{\phi'}), \\  \nonumber
    & \Im(\bra{\phi}\ket{r} \bra{c} \ket{\phi}) \nonumber\\ &= -\frac{i}{2} (\bra{\phi} T^{\dagger} S_k^{\dagger} H_k^{\dagger} H_k S_k T\ket{r} \bra{c} T^{\dagger} S_k^{\dagger} H_k^{\dagger} H_k S_k T \ket{\phi}   \nonumber\\
    &\quad- \bra{\phi}T^{\dagger} S_k^{\dagger} H_k^{\dagger} H_k S_k T\ket{c} \bra{r} T^{\dagger} S_k^{\dagger} H_k^{\dagger} H_k S_k T\ket{\phi})  \nonumber\\
    &= \frac{1}{2} (\bra{\phi''}\ket{r} \bra{r'} \ket{\phi''} + \bra{\phi''} \ket{r'} \bra{r}  \ket{\phi''}), \\ \nonumber
\end{align}
with $\ket{\phi'} = H_kT\ket{\phi}$, and $\ket{\phi''} = H_kS_kT\ket{\phi}$.
Because applying a phase gate $S_k$ doesn't alter the effect of transformation $T$, the extended Bell measurement (XBM) method presented by Kondo \emph{et al.}~\cite{Kondo2022Quantum} can be regarded as analogous to the direct measurement approach used in the current work. In our universal strategy, the gates linked to transformation $T$ are applied initially, potentially followed by the $S$ gate if needed, while in the XBM approach, the $S$ gate is applied first.

\section{Complexity}\label{sec:4_complexity}
The proposed method is implemented in Algorithm~\ref{alg:T} and~\ref{alg:G}. The current work dramatically reduces the number of unique circuits required to estimate expectations defined in Sec.~\ref{sec:2_expe_esti}. In this section, the number of unique circuits and the computational cost of the proposed method are discussed in detail.

\subsection{Circuit Counts}
As previously mentioned in Sec~\ref{sec:circ}, evaluating the fully dense $2^n \times 2^n$ matrix requires $2^n-1$  circuits for computing off-diagonal elements and an additional circuit without gates for direct measurement to assess diagonal elements. This totals $2^n$ unique circuits needed for the computation.

For most engineering problems, linear systems commonly feature matrices with numerous non-zero elements. The performance of current work on such cases is of more practical importance. The following discussion will firstly introduce the concept of band-width, and then discuss circuit counts, gate counts and algorithm complexity.
To facilitate the discussion, the band-width $w$ of the matrix $M$ are defined as
\begin{definition}
    A matrix $M \in{\mathbb C}^{2^n \times 2^n}$ has band-width $w$ if for all $i, j$ satisfying $|i-j|>w$, we have $M_{ij}=0$.
\end{definition}

Based on the definition of  band-width, the theorem (whose proof is in~\cite{Kondo2022Quantum}) holds.

\begin{theorem}[Worst case for bandwidth $w$ matrix]
    \label{theo:2}
    Let $M\in{\mathbb C}^{2^n\times 2^n}$ and $\ket{\phi}, \ket{\psi}\in{\mathbb C}^{2^n}$. When M has the bandwidth $w>0$, the worst case to evaluate $\bra{\phi}|M\ket{\psi}$ needs $2((n-r)k+2^r)$ unique circuits, where $r=\lceil \log_2k \rceil$.
\end{theorem}
The upper bound of the unique circuit counts for an $n$-qubit system with bandwidth $k$ ~\cite{Kondo2022Quantum} is

\begin{align}
\label{eq:upper-m}
\overline{m}(n,k)
&:=\left\{
        \begin{array}{ll}
            2((n-\log_2k)k+2^{\log_2k})  & k>0, \\[5mm]
            1 & k=0.
        \end{array}
    \right.
\end{align}
Readers can refer to the Appendix in~\cite{Kondo2022Quantum} for more details.

\subsection{Gate Counts}
As described in Algorithm~\ref{alg:G}, $CNOT$ gates are applied only when there is an off-diagonal element of value `1' within the $i$-th row of $T$, denoted as $T[i,j]$ where $i \neq j$. Consequently, the maximum number of $CNOT$ gates is limited to $n-1$. Additionally, no more than 2 gates are required for PPM or Direct measurement. Thus, the circuit's gate count is $O(n)$.

\subsection{Time Complexity of grouping for $T$}
Determining the $(r,c)$  pairs for a specific transformation in $M$ involves a classical pre-processing step. This step primarily revolves around binary addition, i.e. $AND$, on $n$ bits, with a computational cost of $O(n)$. If matrix $M$ as $\{M_{rc}\ne0\}$ contains $p$ non-zero elements, the time complexity for grouping terms is $O(np)$. Constructing each measurement circuit requires a time complexity of $O(n)$. It involves a maximum of $n-1$ $CNOT$ gates and, at most, an additional three gates which can be $X$ , $H$ or $S$. More importantly, each circuit gets constructed only once for each unprocessed value of $r + c$. If there are $d$ distinct $r + c$ value, the overall time complexity for building the circuits becomes $O(nd)$.

\section{Results and Discussion}
\label{sec:5_results}
\begin{figure*}
    \centering
    \begin{subfigure}[t]{0.33\textwidth}
        \centering
        \includegraphics[width=0.9\textwidth]{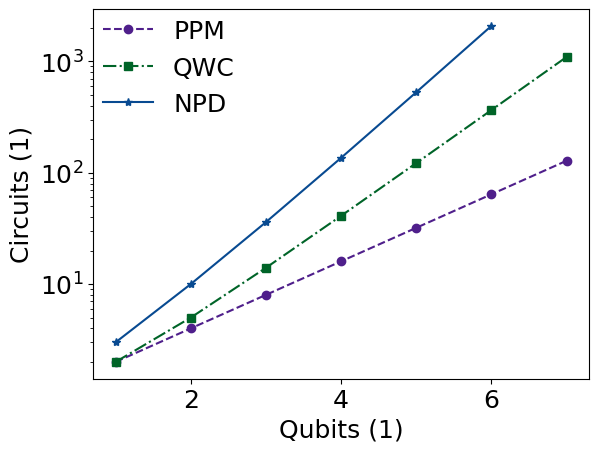}
        \caption{}
    \end{subfigure}%
    \begin{subfigure}[t]{0.33\textwidth}
        \centering
        \includegraphics[width=0.9\textwidth]{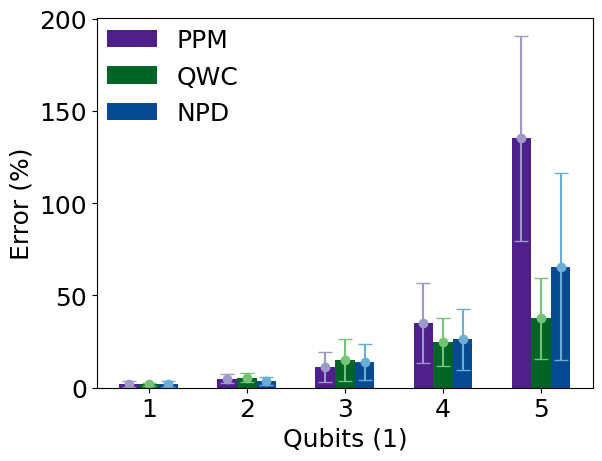}
        \caption{}
    \end{subfigure}%
    \begin{subfigure}[t]{0.33\textwidth}
        \centering
        \includegraphics[width=0.9\textwidth]{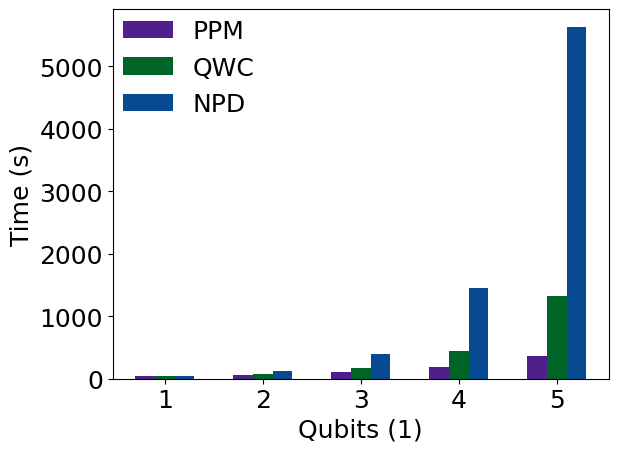}
        \caption{}
    \end{subfigure}%

    \begin{subfigure}[t]{0.33\textwidth}
        \centering
        \includegraphics[width=0.9\textwidth]{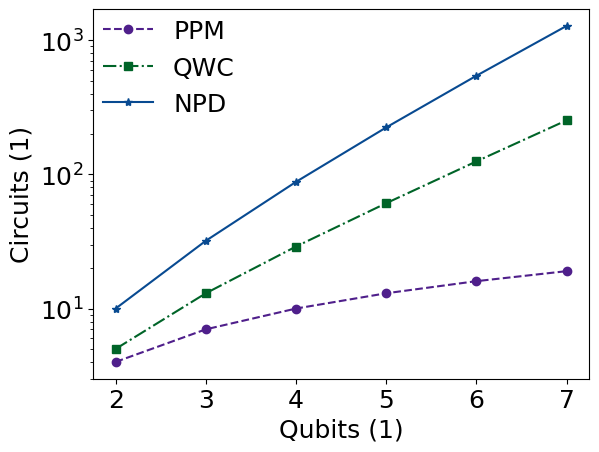}
        \caption{}
    \end{subfigure}%
    \begin{subfigure}[t]{0.33\textwidth}
        \centering
        \includegraphics[width=0.9\textwidth]{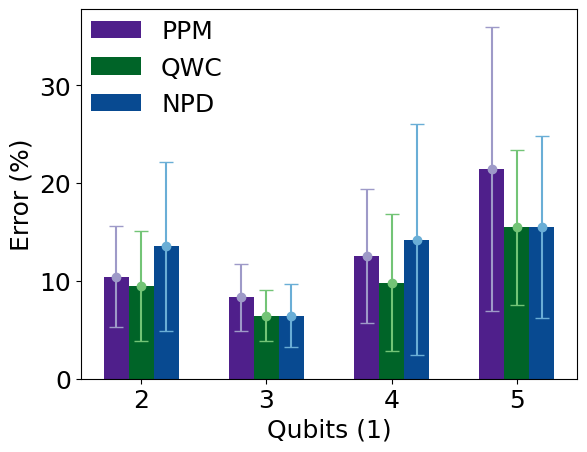}
        \caption{}
    \end{subfigure}%
    \begin{subfigure}[t]{0.33\textwidth}
        \centering
        \includegraphics[width=0.9\textwidth]{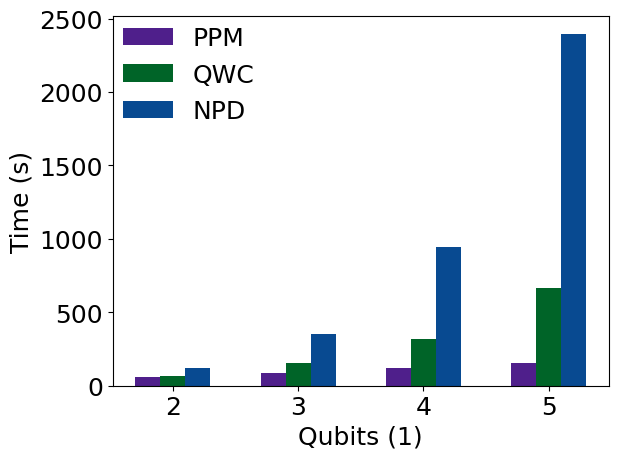}
        \caption{}
    \end{subfigure}%

    \begin{subfigure}[t]{0.33\textwidth}
        \centering
        \includegraphics[width=0.9\textwidth]{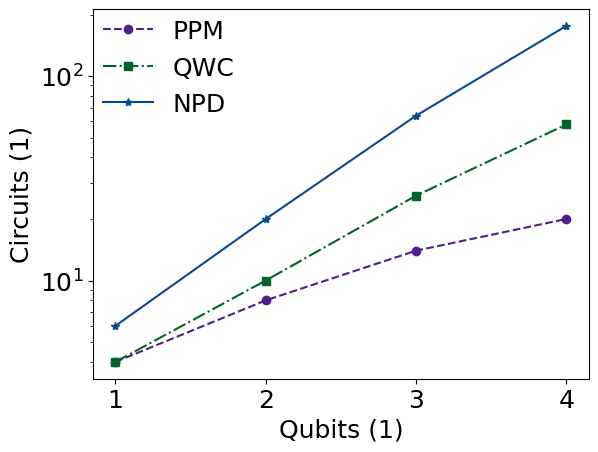}
        \caption{}
    \end{subfigure}%
    \begin{subfigure}[t]{0.33\textwidth}
        \centering
        \includegraphics[width=0.9\textwidth]{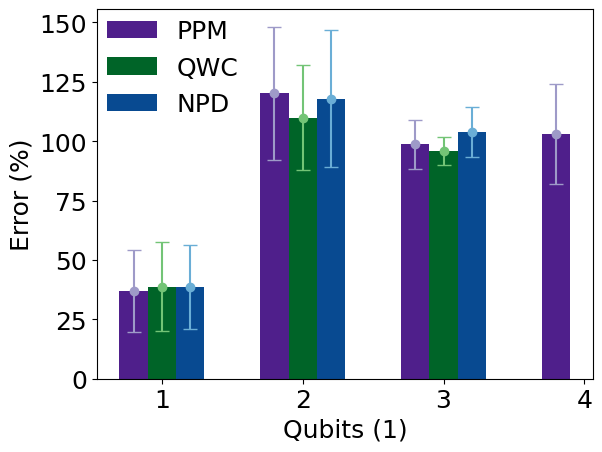}
        \caption{}
    \end{subfigure}%
    \begin{subfigure}[t]{0.33\textwidth}
        \centering
        \includegraphics[width=0.9\textwidth]{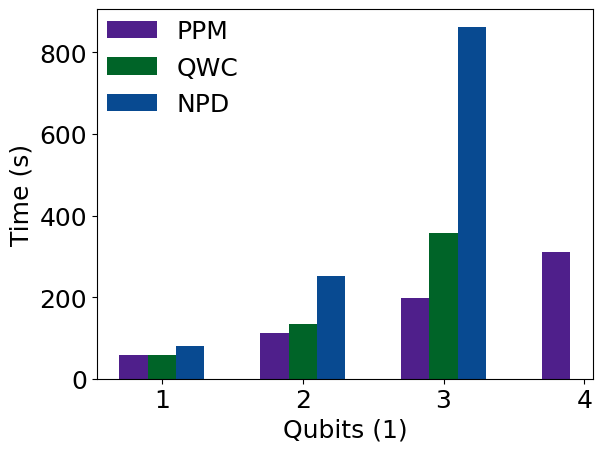}
        \caption{}
    \end{subfigure}%
    \caption{
    (a, b, c), (d, e, f), (g, h, i) the number of distinct circuits, the error of real hardware results compared with theoretical value, the execution time on the real hardware to estimate $\bra{\phi} M \ket{\psi}$ with 1) $\bra{\phi} M \ket{\phi}$ with $M$ a symmetric full matrix, 2) $\bra{\phi} M \ket{\phi}$ with $M$ a symmetric banded matrix (bandwidth 3), and 3) $\bra{\phi} M \ket{\psi}$ with $M$ a symmetric banded matrix (bandwidth 3).
PPM:the proposed method; NPD: Na\"ive Pauli Decomposition measurements; QWC: qubit-wise commuting Pauli measurements~\citep{Verteletskyi2020TJCP, McClean2016NJP}.}
\label{fig:results}
\end{figure*}

We implemented the proposed partial Pauli Decomposition method (PPM) and conducted tests using Qiskit version 1.0.2~\cite{qiskit} on the IBM-Q \texttt{ibm\_hanoi} quantum hardware. To benchmark the performance, we compared PPM against two existing methods: the Na\"ive Pauli Decomposition (NPD) method and the qubit-wise commuting (QWC) method, implemented as \texttt{SparsePauliOp} and \texttt{SparsePauliOp.group\_commuting(qubit-wise=True)}, respectively, in Qiskit.

For simplicity, we focused our tests on symmetric matrices $M$, as the anti-symmetric case only adds a single $S$ gate with negligible impact. We evaluated three test cases: 1) $\bra{\phi} M \ket{\phi}$ with $M$ a symmetric full matrix, 2) $\bra{\phi} M \ket{\phi}$ with $M$ a symmetric banded matrix (bandwidth 3), and 3) $\bra{\phi} M \ket{\psi}$ with $M$ a symmetric banded matrix (bandwidth 3). The matrix elements were randomly generated, and the error was calculated by comparing the hardware results to those from the statevector simulator. We repeated each test case 10 times and report the mean and standard deviation.

Figure~\ref{fig:results} shows the number of unique circuits, the hardware error, and the execution time for the three test cases. For the full matrix case (a-c), PPM requires exponentially fewer circuits than NPD and QWC, scaling as $2^n$ compared to $2^{1.52n}$ and $2^{1.89n}$, respectively. While PPM's error is comparable to the other methods for $n<4$, it worsens for larger $n$ due to the additional CNOT gates introduced to reduce the circuit count. Nevertheless, PPM's execution time is significantly shorter.

For the banded matrix case (d-f), PPM's circuit count advantage over NPD and QWC becomes more pronounced as $n$ increases, requiring orders of magnitude fewer circuits for large $n$. Remarkably, for $n=7$, PPM used only 19 circuits compared to 253 for QWC and 1280 for NPD. The hardware errors are comparable across methods at around 15\%, with a minor degradation for PPM due to the CNOT gates. Despite this, PPM maintains a substantial execution time advantage.

In the $\bra{\phi} M \ket{\psi}$ case (g-i), while PPM retains its circuit count efficiency, the errors for all methods escalate due to the large number of two-qubit gates required to construct $\ket{\varphi} = (\ket{0}\ket{\phi} + \ket{1}\ket{\psi})/\sqrt{2}$. Mitigating this error would require reducing the two-qubit gate error rates.

In summary, the PPM method enables significantly faster evaluation of $\bra{\phi} M \ket{\phi}$ for both full and banded matrices, with a minor trade-off in accuracy due to the additional CNOT gates introduced. Further improvements in two-qubit gate fidelities would enhance PPM's performance for more general cases like $\bra{\phi} M \ket{\psi}$. Overall, PPM presents a promising approach for efficient quantum computations on near-term devices.

\section{conclusions}\label{sec:6_conclusion}
In conclusion, the method we introduce is capable of handling any $2^n \times 2^n$ matrix with at most $2^n$ distinct quantum circuits. It does so by leveraging the binary arithmetic properties of X and CNOT gates to achieve computational efficiency, especially for small band-width matrices.
This approach not only simplifies the process of measuring matrix elements but also enhances the feasibility of executing complex quantum algorithms in the NISQ era.
Our findings are positioned to serve as a cornerstone for future advancements, especially in algorithms reliant on expectation estimation like VQE and QAOA.

\begin{acknowledgments}
This research is supported by the National Research Foundation, Singapore and A*STAR under its Quantum Engineering Programme (NRF2021-QEP2-02-P04).
\end{acknowledgments}

\appendix\label{sec:7_app}
\setcounter{figure}{0}
\renewcommand{\thefigure}{A\arabic{figure}}

\section{Proofs of Theorems}\label{app:prov}
\setcounter{theorem}{0}
\begin{theorem}[\textbf{Non-Overlap}]
Let
$i, j, k, l, p, p', q, q'\in \mathbb Z_2^n$, $p \neq p'$, $q \neq q'$,
$T_p, T_q \in \mathbb Z_2^{n \times n}$,
$T_p p = p'$,
$T_q q = q'$,
$S_p = \big \{ (i, j):  T_p (i, j)^T = (i, i')^T, i \oplus i' = p' \big \}$,
$S_q = \big \{ (k, l):  T_q (k, l)^T = (k, k')^T, k \oplus k' = q' \big \}$.
Then, $S_p \cap S_q = \emptyset $
if $p \neq q$.
\end{theorem}

\begin{proof}
\begin{subequations}
\begin{eqnarray}
    T_p i  = i, \label{eq:trans_rc1} \\
    T_p j  = i'. \label{eq:trans_rc2}
\end{eqnarray}
\end{subequations}
Adding Eq.~\ref{eq:trans_rc1} and Eq.~\ref{eq:trans_rc2} yields
\begin{eqnarray}
    T_p(i \oplus j) = i \oplus i' = p'. \label{eq:hp_solu}
\end{eqnarray}

Recall $T_p$ is defined by $p$ as $T_p p = p'$. We have
\begin{eqnarray}
    T_p(i \oplus j) = T_p p,
 \end{eqnarray}
\begin{eqnarray}
    p = i \oplus j,
\end{eqnarray}

Similarly,
\begin{eqnarray}
    q = k \oplus l.
\end{eqnarray}

Since $p \neq q$, the pairs of $(i, j)$ and $(k, l)$ do not have common elements, $S_p \cap S_q = \emptyset$.
\end{proof}

\section{$\ket{\phi} \neq \ket{\psi}$}
\label{app:two_states}
For the case $\ket{\phi} \neq \ket{\psi}$, the expectation $\bra{\phi} M \ket{\psi}$ can be evaluated by introducing an ancilla qubit.
Let $M \in {\mathbb C}^{2^{n} \times 2^{n}}$ be a complex matrix and $\ket{\phi}, \ket{\psi}\in{\mathbb{C}^{2^{n}}}$ be two $n$-qubit quantum states.
The expectation $\bra{\phi} M \ket{\psi}$ can be expressed as
\begin{eqnarray}
\bra{\phi} M \ket{\psi} = \bra{\varphi} M' \ket{\varphi},
\end{eqnarray}
where
\begin{eqnarray}
M' = \begin{cases}
        M & {\rm if}\ \ \ \ket{\phi} = \ket{\psi}, \\ \\
        \begin{bmatrix} {\bf 0} & 2M \\ {\bf 0} & {\bf 0} \end{bmatrix} & {\rm if}\ \ \ \ket{\phi} \neq \ket{\psi},
    \end{cases}
\end{eqnarray}
and
\begin{eqnarray}
\ket{\varphi} = \begin{cases}
        \ket{\phi} & {\rm if}\ \ \ \ket{\phi} = \ket{\psi}, \\ \\
        \dfrac{\ket{0}\ket{\phi} + \ket{1}\ket{\psi}}{\sqrt{2}} & {\rm if}\ \ \ \ket{\phi} \neq \ket{\psi}.
    \end{cases}
\end{eqnarray}
Note that the state $\ket{\varphi} = \dfrac{\ket{0}\ket{\phi} + \ket{1}\ket{\psi}}{\sqrt{2}}$ can be prepared using the circuit described in Fig.~\ref{fig:221}(c)~\cite{Liu2021PRA}.

\begin{figure}[hbt!]
\centering
\begin{equation*}
\Qcircuit @C=1.em @R=0.4em @!R {
    \lstick{0^n}   & \qw      &  \gate{M_{\ket{\phi}}} & \qw      &  \gate{M_{\ket{\psi}}}  & \qw \\
    \lstick{0}     & \gate{H} & \ctrl{-1}              & \gate{X} & \ctrl{-1}               & \qw \\
 }
\end{equation*}
\caption{The circuit transforms the estimation of expectations from two different states onto a common state by utilizing an ancilla qubit.}
\label{fig:221}
\end{figure}
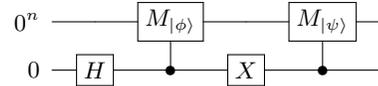

\section{Procedure for obtaining the binary transformation matrix $T$}\label{app:TG}
As stated in Sec.~\ref{bina_arit}, the transformation $T$ can convert the pairs $(r, c)$ to $(r, r')$.
The matrix $T$ is determined from
\begin{subequations}
\begin{eqnarray}
    r &=& T r \oplus R, \label{eq:trans1}\\ 
    c &=& T r' \oplus R,
    \label{eq:trans2}
\end{eqnarray}
\end{subequations}
where $T$ must be unitary since $T$ is the combined effect of $X$ and CNOT gates. From Eq.~\ref{eq:trans1}, it's evident that the constant binary vector $R=0$. By adding Eq.~\ref{eq:trans1} and Eq.~\ref{eq:trans2} together,
\begin{eqnarray}
    r \oplus c &=& T (r \oplus r').
\end{eqnarray}
Here $r \oplus r'$ is a special vector with zeros in all elements except for the $k$-th element as 1. It reflects the $k$-th column of $T_k = [q_{1k}, q_{2k}, \ldots q_{nk}]^T$ equals to $r \oplus c$.  Hence the matrix $T$ can be easily constructed by replacing the $k$-th column of an $n\times n$ Identity matrix with $r \oplus c$. It is noted that the $k$-th bits of $r$, $c$ and $r \oplus c$ is 0, 1, 1 such that

\begin{eqnarray}
    T = I \oplus A \oplus B,
\end{eqnarray}
where $A$ is an $n \times n$ matrix, with all columns as zeros except for the $k$-th column as $r \oplus c$, and $B$ an $n \times n$ matrix with all zeros except for the single element at row $k$ and column $k$ being 1. They are given by
\begin{align}
    A &= [\textbf{0}, \textbf{0}, \cdots, \textbf{0}, \textbf{r $\oplus$ c} , \textbf{0}, \cdots, \textbf{0}, \textbf{0}], \nonumber\\
    B &= [\textbf{0}, \textbf{0}, \cdots, \textbf{0}, \textbf{v} , \textbf{0}, \cdots, \textbf{0}, \textbf{0}],
\end{align}
where $\textbf{0}$ is an $n \times 1$ vector of all zeros and $\textbf{v}$ is an $n \times 1$ vector of all zeros except for the $k$-th element being 1.
\begin{align}
    \textbf{0} = \begin{bmatrix}
        0\\
        0\\
        \vdots\\
        0\\
        0\\
    \end{bmatrix},
    \textbf{v} = \begin{bmatrix}
        0\\
        \vdots\\
        0\\
        1\\
        0\\
        \vdots\\
        0\\
    \end{bmatrix}.
\end{align}
Since the $k$-th bit of $r$ is 0, $A r = \textbf{0}$ and $B r = \textbf{0}$. On the other hand, since the $k$-th bit of $c$ is 1, $A c = r \oplus c$ and $B c = \textbf{v}$. Hence the transformation of $(r, c)$ to $(r, r')$ using $T$ can be verified as
\begin{align}
T r = (I \oplus A \oplus B)r = r \oplus \textbf{0} \oplus \textbf{0} = r, \\
T c = (I \oplus A \oplus B)c = c \oplus r \oplus c \oplus \textbf{v} = r \oplus \textbf{v} = r'.
\end{align}

The Algorithm~\ref{alg:T} details how to construct $T$,
\begin{algorithm*}[tb]
    \caption{Algorithm for determining bit linear transformation matrix $T$ from $i$}
    \label{alg:T}
    \SetKwInOut{Input}{input}
    \SetKwInOut{Output}{output}
    \SetKwInOut{Initialize}{initialize}
    \Input{Number of qubits $N$, integer $r, c$, with $N$ bits binary representation of $r = q_{N-1}, \ldots, q_0$,  $c = p_{N-1}, \ldots, p_0$ and $r \neq c$ }
    \Output{The $N \times N$ bit linear transformation matrix $T$}
    \Initialize{$T = I$, $I$ is $N \times N$ Identity matrix, $v = r \oplus c$, $u$ list of set bits in $r\oplus c$, $u \in \{[0, \ldots, N-1]|v_{u} =1\}$ and  $k = u[0]$}

\For{$i\leftarrow 0$ \KwTo $N-1$}{
      $T[k,i]=v[i]$
      }
\end{algorithm*}

As discussed in Sec.~\ref{bina_arit}, it is straightforward to retrieve the sequence of CNOT gates from the transformation matrix $T$. Specifically,the first step is to check each row of $T$ denoted as $T[i, :]$. It aims to find whether there is another off-diagonal element with value 1 in the $j$-th column, i.e., $T[i, j] = 1$ where $i \neq j$. If there exists such a column $j$, a CNOT gate will be introduced into the gate list where $j$ works as control and $i$ as target. The Algorithm~\ref{alg:G} is proposed to achieve this objective.
\begin{algorithm*}[tb]
    \caption{Algorithm for Determining the circuit from the bit linear transformation matrix $T$}
    \label{alg:G}
    \SetKwInOut{Input}{input}
    \SetKwInOut{Output}{output}
    \SetKwInOut{Initialize}{initialize}
    \Input{The $N \times N$ bit linear transformation matrix $T$}
    \Output{The list of X and CNOT gates $G$ for the measurement.}
    \Initialize{$G\leftarrow\texttt{None}$}
\For{$i\leftarrow 0$ \KwTo $N-1$}{
    $k$, column indices where $T[i,k] = 1, k \neq i$\\
    \uIf{$k$}{
        $G \leftarrow G + CNOT(k, i)$\\
    }
}
\end{algorithm*}

%
\end{document}